\documentclass[conference]{IEEEtran}
\usepackage{amsmath,amsthm,amssymb,latexsym,ifthen}

\newtheorem{thm}{Theorem}[section]
\newtheorem{lem}[thm]{Lemma}
\newtheorem{cor}[thm]{Corollary}
\newtheorem{prop}[thm]{Proposition}

\theoremstyle{definition}
\newtheorem{defn}[thm]{Definition}

\theoremstyle{remark}

\newtheorem*{clm}{Claim}

\renewcommand{\leq}{\leqslant}
\renewcommand{\geq}{\geqslant}

\ifCLASSINFOpdf
\else
\fi
\begin{document}
%
\title{Large Scale Geometries of Infinite Strings}

\author{\IEEEauthorblockN{Bakh Khoussainov}
\IEEEauthorblockA{Department of Computer Science\\
The University of Auckland\\
Auckland, New Zealand\\
Email: bmk@cs.auckland.ac.nz}
\and
\IEEEauthorblockN{Toru Takisaka}
\IEEEauthorblockA{Research Institute for Mathematical Sciences\\
Kyoto University\\
Kyoto, Japan\\
Email: takisaka@kurims.kyoto-u.ac.jp}
}


%


\IEEEoverridecommandlockouts
\IEEEpubid{\makebox[\columnwidth]{978-1-5090-3018-7/17/\$31.00~
\copyright2017 IEEE \hfill} \hspace{\columnsep}\makebox[\columnwidth]{ }}

\maketitle

\begin{abstract}
We introduce geometric consideration into 
the theory of formal languages. We aim  to shed light on our understanding of global patterns 
that occur on infinite strings. We utilise methods of geometric group theory.
Our emphasis is on  large 
scale geometries. 
Two infinite  strings have the same large scale
geometry if there are colour preserving bi-Lipschitz maps with  
distortions between the strings. 
Call these maps
quasi-isometries. Introduction of  large scale geometries poses several
questions. The first question asks to study the partial order induced by  quasi-isometries.   
This partial order compares large scale geometries; as such  it presents  
an algebraic tool for classification of  global patterns. We prove there is a greatest large scale geometry  
and infinitely many minimal large scale geometries.
The second question is related to understanding the quasi-isometric
maps on various classes of strings.
The third question investigates the sets of large
scale geometries of strings accepted by 
 computational models, e.g.  B\"uchi automata. 
 We provide an algorithm that describes  large scale
geometries of strings accepted by B\"uchi automata. This links
large scale geometries with automata  theory.
The fourth question studies  the complexity of the quasi-isometry problem.
We show the 
problem is $\Sigma_3^0$-complete thus providing a bridge
with computability theory.
Finally, the fifth question asks to
build  algebraic
structures that are invariants of large scale geometries. 
We invoke asymptotic cones, a key concept in geometric group theory, 
defined via model-theoretic notion of ultra-product. Partly, we study asymptotic cones of algorithmically random strings
thus connecting the topic with algorithmic randomness.
\end{abstract}


%
\IEEEpeerreviewmaketitle

\section{Introduction}

Our goal is to introduce geometric considerations into the theory of formal languages.
We  emphasise  the study of large scale geometries of infinite strings. 
Our hope is to shed light on our understanding of global large scale patterns that occur on infinite strings.
Our motivation comes from geometric group theory, a cutting edge research area in group theory 
linked with 
geometry, topology, automata, logic, probability, complexity, etc. 


In geometric group theory, an important 
concept is that of 
 quasi-isometry between metric spaces.
Let $\mathcal M_1=(M_1, d_1)$ and $\mathcal M_2=(M_2, d_2)$ be metric spaces. 

\vspace{-1mm}

\begin{defn} \label{Defn:QIM}
A function $f:M_1\rightarrow M_2$ is called an {\bf $(A,B)-$quasi-isometry} from $\mathcal M_1$ to $\mathcal M_2$, where $A\geq 1$ and $B\geq 0$, if for all $x,y\in M_1$ we have 
$$
(1/A) \cdot d_1(x,y) -B \leq d_2(f(x), f(y))\leq A \cdot d_1(x,y) +B,
$$
and for all $y\in M_2$ there exists an element  $x\in M_1$ such that $d_2(y, f(x))\leq A$. 
\end{defn}

\vspace{-1mm}

Note that when $B=0$, the mapping becomes bi-Lipschitz (and hence continuous).
Thus, a quasi-isometry $f: \mathcal M_1\rightarrow \mathcal M_2$ behaves like a bi-Lipschitz map with distortion $B$ between the metric spaces. For instance, the metric spaces  $\mathbb Z$ (integers) and  $\mathbb R$ (reals) with their natural metric are quasi-isometric  Informally, two metric spaces $\mathcal M_1$ and $\mathcal M_2$ are quasi-isometric if these spaces (such as  $\mathbb Z$ and  $\mathbb R$) look the same from far away. 
The quasi-isometry relation forms 
an equivalence relation on the class of all metric spaces.



Studying quasi-isometry invariants of groups turned out to be crucial in solving many problems in group theory \cite{Gromov1} \cite{Gromov2} \cite{Gromov3}. 
Therefore, finding quasi-isometry invariants  has become an important theme in geometric group theory. 
Examples of quasi-isometry invariants are: being virtually nilpotent, virtually free,  hyperbolic, having polynomial growth rate, 
being finitely presentable, having decidable word problem, asymptotic cones  \cite{Dries} \cite{DrutuKap} \cite{Gromov3}. 


In  formal language theory and logic, one of the main  objects are infinite strings $\alpha$ 
over finite alphabets $\Sigma$. These objects are somewhat boring from a geometry point of view. The strings $\alpha$ possess  the natural metric inherited from the set of natural numbers $\omega$. 
The quasi-isometry type of $\omega$ is the metric space $\mathbb R_{\geq 0}$ of all positive reals. 
So, from a quasi-isometry view point, $\omega$ viewed as a metric space is also somewhat uneventful. 
However, one crucial difference from the setting in geometric group theory  is that the metric spaces $\alpha$ are coloured. Namely,  for every position $i\in \omega$ in $\alpha$, 
the colour of the position $i$ is $\sigma$ 
when $\alpha(i)=\sigma$.
These observations suggest that 
the notions and methods of geometric group theory (e.g. quasi-isometry) can be applied to coloured metric spaces. Here we investigate quasi-isometries of  coloured metric spaces with strings $\alpha$ as our primary objects, thus initiating the study of large scale patterns on strings.


A {\bf coloured metric space} $\mathcal M$ is a tuple $(M; d, C)$, where $(M;d)$ is the {\bf underlying 
metric space} with metric $d$, $C$ is a colour function $C:M \rightarrow 2^{\Sigma}$, and $\Sigma$ is a finite set of colours that we call an alphabet. If $\sigma \in C(m) $ then we say that {\em $m$ has colour $\sigma$}. 
As mentioned above, infinite strings $\alpha$ 
are  coloured metric spaces of the form  
$(\omega; d, \alpha)$, \ 
where $d$ is the natural metric (so, $d(i,j)=|i-j|$) and 
$\alpha: \omega \rightarrow \Sigma$ is the colour function.  Every element in $\alpha$ has a unique colour.
We always assume that the cardinality $|\Sigma|$ of $\Sigma$ is at least $2$.



We denote the set of all infinite strings over $\Sigma$, considered as coloured metric spaces, by $\Sigma^{\omega}$.  
We now adapt Definition \ref{Defn:QIM} for coloured metric spaces.

\begin{defn}
Assume we are given coloured metric spaces $\mathcal M_1=(M_1; d_1, C_1)$ and $\mathcal M_2=(M_2; d_2, C_2)$. 
A {\em colour preserving} quasi-isometry from $(M_1; d_1)$ into $(M_2; d_2)$ is called a
{\bf quasi-isometry} from $\mathcal M_1$ into $\mathcal M_2$.
\end{defn}

We write  $\mathcal M_1\leq_{QI} \mathcal M_2$ if there is a quasi-isometry  from $\mathcal M_1$ into $\mathcal M_2$. 
For metric spaces the relation $\leq_{QI}$ is an equivalence relation, and hence it is symmetric. In contrast, for 
coloured metric spaces the relation $\leq_{QI}$ is not symmetric.
For instance, the mapping $f: \omega\rightarrow \omega$ defined as $f(i)=2i$ is a quasi-isometry from the string $0^{\omega}=00000\ldots$ into  the string $(01)^{\omega}=010101\ldots$.
There is no quasi-isometric mapping in the opposite direction.


If $\mathcal M_1 \leq_{QI} \mathcal M_2$  and $\mathcal M_2 \leq_{QI} \mathcal M_1$ then we write  
this by $\mathcal M_1\sim_{QI} \mathcal M_2$.
The relation $\sim_{QI}$ is an equivalence relation in the class of all coloured metric spaces, in particular on the set $\Sigma^{\omega}$. 
Formally:

\begin{defn} \label{Dfn:QI-type}
The equivalence classes of $\sim_{QI}$ are called  {\bf the quasi-isometry types} or, equivalently,  {\bf the large scale geometries}. For the set $\Sigma^{\omega}$, we define 
$\Sigma_{QI}^{\omega}=\Sigma^{\omega}/\sim_{QI}$. 
Denote the large scale geometry of string $\alpha$  by $[\alpha]$. Thus, $\Sigma_{QI}^{\omega}=\{[\alpha]\mid \alpha \in \Sigma^{\omega}\}$.
\end{defn}

 Non-symmetry of $\leq_{QI}$ on coloured metric spaces allows 
us to compare large scale geometries of these spaces, and consider the partial order $\leq_{QI}$ on the quasi-isometry 
types. Importantly, the partial order $\leq_{QI}$ can be restricted to the large scale geometries of coloured metric spaces over a fixed  underlying metric space (e.g. $\omega$ with its natural metric). In this sense the quasi-isometry on coloured metric spaces is a much refined version of the usual quasi-isometry relation on unclouded metric space. This is because the quasi-isometry type of every uncoloured 
metric space trivialises  to a singleton. 


Introduction of large scale geometries  and the quasi-isometry relation $\leq_{QI}$ poses a vast amount of natural questions. The contribution of this  paper consists of investigating  the following questions:


{\em The first  question} is  concerned with understanding the partial order $\leq_{QI}$ on the set  
$\Sigma_{QI}^{\omega}$ of all large scale geometries.  The order presents an algebraic tool aimed at classification of global patterns that occur on strings. In Section \ref{S:order} we investigate properties of this partial
order. Among several results, we prove that the order has the greatest large scale geometry, infinite chains and infinite 
anti-chains. We show that there are infinitely many minimal
large scale geometries.

 {\em The second question} is related to understanding 
the quasi-isometry relation $\sim_{QI}$. This can be done by either restricting the relation $\sim_{QI}$ on various subclasses  of infinite strings or by 
simplifying the definition of quasi-isometry. In Section \ref{S:EP} we restrict the relation $\sim_{QI}$ to eventually periodic words,
and give a full description of quasi-isometry types in this class. Section \ref{S:Colour-equivalence} gives a
intuitively more easier yet equivalent definition  of the relation
$\leq_{QI}$ that we call component-wise reducibility. We also give more refined version of $\leq_{QI}$,
colour-equivalence, that
implies quasi-isometry. It is shown that colour-equivalence  is strictly stronger than quasi-isometry. It is natural to ask if quasi-isometric maps between strings can be replaced with order preserving quasi-isometric embeddings. 
We give a negative answer to this question; however, we prove that the answer is positive modulo $\sim_{QI}$ equivalence relation.

 {\em The third question} is related to describing sets of large scale geometries. We call such sets atlases. Let $L$ be a language of infinite strings. One 
 considers the atlas $[L]$ of all large scale geometries of strings in $L$. 
So, the  question is related to understanding the atlas $[L]$ given a description of $L$. In particular, a natural question is if one can decide that $[L_1]=[L_2]$ given descriptions of languages $L_1$ and $L_2$. In Section \ref{S:Buchi}, we use B\"uchi automata as a description language and provide a full characterisation of
the atlases defined by B\"uchi recognisable languages. As a consequence, we show that for B\"uchi automata recognisable languages $L_1$ and $L_2$ there is a linear time algorithm 
for deciding the equality of the atlases $[L_1]$ and $[L_2]$.  This contrasts with the PSPACE completeness of  the equality problem for B\"uchi languages. This part of the work links large scale geometries with automata theory and complexity theory.

 {\em The fourth question} addresses the complexity of quasi-isometry relation on infinite strings in terms of arithmetical hierarchy,
 thus connecting the topic with computability theory. 
 Recall that isometry is colour preserving and distance preserving map. So,  quasi-isometry relation is weaker than isometry relation. 
Hence, one  might expect that quasi-isometry is easier to detect than the isometry. We prove that the quasi-isometry relation on computable infinite strings is $\Sigma_3^0$-complete. This is in contrast to $\Pi_1^0$-completeness of 
the isometry relation on computable strings. These are explained in Section \ref{S:complexity}.

 {\em The fifth question} asks how one encodes large scale geometries of coloured metric spaces $\alpha$ 
into ``limit" structures.  
For this, we define structures obtained from ``looking at $\alpha$ from far away". 
 Dries and Wilkie  \cite{Dries}, using ultra-filters, 
 formalised this intuitive notion through the concept of asymptotic cone. Their work gave a logical context to Gromov's work in \cite{Gromov1} \cite{Gromov3}.
We invoke the concept of asymptotic cone 
and study relationship between large scale geometries of infinite strings and their asymptotic cones. In Section \ref{Cones} we prove theorems akin to results on asymptotic cones in geometric group theory, and we show that asymptotic cones of Martin-L\"of random strings coincide when scaling factors are computable. These results bridge the topic of this paper with algorithmic randomness and model theory.


\vspace{-1mm}

\section{The partial order $(\Sigma_{QI}^{\omega}, \leq_{QI})$}\label{S:order}

\subsection{Basic Properties}

\vspace{-1mm}

The quasi-isometry relation $\leq_{QI}$ naturally induces the partially ordered set $(\Sigma_{QI}^{\omega}, \leq_{QI})$ 
on the set of all large scale geometries $[\alpha]$.   Denote this partial order by $\Sigma_{QI}^{\omega}$ thus identifying it with its domain.  
Say that $[\beta]$ {\em covers} $[\alpha]$ if $[\alpha]\neq [\beta]$, $[\alpha]\leq_{QI} [\beta]$,  and for all  $[\gamma]$ 
such that $[\alpha]\leq_{QI} [\gamma] \leq_{QI} [\beta]$ we have $[\gamma]=[\alpha]$ or $[\gamma]=[\beta]$. An element is an {\em atom} if it covers a minimal element.


\begin{prop}\label{Prop:elementary}
The partial order $\Sigma_{QI}^{\omega}$ has the following properties:
(1) There exists a greatest element; (2) There exist at least $|\Sigma|$ minimal elements; (3) 
There exist at least $|\Sigma|\cdot (|\Sigma|-1)/2$ atoms.
\end{prop}
\begin{proof} 
Assume that $\Sigma = \{\sigma_1, ... , \sigma_k\}$. 
For part (1),  we claim that  $\alpha = (\sigma_1...\sigma_k)^\omega$ is the greatest element in $\Sigma_{QI}$. Indeed, take any $\beta \in \Sigma^{\omega}$.  We write $\alpha$ as $v_0 v_1 \ldots$ where each $v_i$ is $\sigma_1\ldots \sigma_k$. Define $f:\beta \rightarrow \alpha$ such that $f$ maps any position $n$ (in $\beta$) to
the position $k_n$ in the portion $v_n$ of the string $\alpha$ such that $\beta(n)=\alpha(k_n)$.  
For part (2), consider the quasi-isometry types $[\sigma_{i}^{\omega}]$. These 
form minimal elements in $\Sigma_{QI}^{\omega}$.  Indeed, if $\alpha\leq_{QI} \sigma^{\omega}_i$, then each element of  
$\alpha$ has  colour $\sigma_i$. For part (3) ,  consider $\sigma_i (\sigma_j)^{\omega}$, where $i\neq j$. 
Clearly, $\sigma_j^{\omega}\leq_{QI} \sigma_i (\sigma_j)^{\omega}$. Moreover, for any $\beta\neq \sigma_j^{\omega}$
if $\beta \leq_{QI} \sigma_i (\sigma_j)^{\omega}$ then $\beta$  finitely many positions of $\beta$ are coloured with $\sigma_i$, and all other positions are coloured with $\sigma_j$. Hence, $\beta\sim_{QI} \sigma_i (\sigma_j)^{\omega}$.
Thus, we have found  at least $k\cdot (k-1)/2$  of the atoms.
\end{proof}

A quasi-isometry $g:\alpha\rightarrow \beta$ can produce {\em cross-overs}, e.g., pairs $n$ and $m$ with $n<m$ but $g(m)< g(n)$. The definition of quasi-isometry does not obviously prohibit large cross-overs. Nevertheless, 
the following  lemma shows that there is a uniform bound on cross-overs. 

\begin{lem}[{\bf Small Cross Over Lemma}]\label{L:no-big-crossing}
Consider a quasi-isometry  map $g:\alpha\rightarrow \beta$. There is a constant $C\leq 0$ 
such that for all $n<m$ we have $g(m)-g(n)\geq C$.
\end{lem}
\begin{proof} Let $n< m$ be given and suppose we have 
$g(m)-g(n)<C$ for some $C \leq 0$. The goal is to provide a lower bound for $C$.  
Define $q = \min(g^{-1}(p) \cap [m+1, \infty))$ and 
$p=\min (g([m+1,\infty)) \cap [g(n)+1, \infty))$. \ 
Then, we obtain
\begin{eqnarray*}
d(g(n), p) & \geq & (1/A) \cdot d(n,q)-B \geq  (1/A) d(n,m)-B\\ 
           & \geq & (1/A) (1/A) (d(g(n), g(m))-B)-B\\
           & \geq & -C/A^2 -((A^2+1)/(A^2)) B.
\end{eqnarray*}
We have $g([m,q]) \subset [0,g(n)] \cup [p,\infty)$ by our definition of $p$ and $q$. Note $g(m)\in[0,g(n)]$ and $g(q)\in[g(p),\infty)$. Hence, there exists $r_C \in [m, q-1]$ such that $g(r_C)\in[0,g(n)]$ and $g(r_C+1)\in[g(p),\infty)$, which means 
$$
d(g(n),p) \leq d(g(r_C),g(r_C+1)) \leq A + B.
$$ 
From these 
inequalities we have a lower bound on $C$. 
\end{proof}

\vspace{-2mm}
\begin{lem}\label{L:pre-image}
Let $f:\alpha \rightarrow \beta$ be a quasi-isometric mapping. Then there exists constants $A'$ and $B'$ such that
for all positions $x, y\in \beta$ with $x\leq y$ we have 
\begin{center}{
$(1/A') \cdot d_2(x,y) -B'  \leq $  \\
$d_1(\min f^{-1}([x,y]),\max f^{-1}([x,y]))  \leq 
 A'  \cdot d_2(x,y) +B'$.}
\end{center}
\end{lem}
\begin{proof} 
The inequality in the definition of quasi-isometry implies:
\begin{center}{
$(1/A)\cdot d_2(f(n),f(m))-(B/A)\leq d_1(n,m)$\\
$\leq A\cdot d_2(f(n),f(m))+AB.$}
\end{center}
For given $x\leq y\in \beta$, let $M_{xy}= \max f^{-1}([x,y])$ and $m_{xy}= \min f^{-1}([x,y])$. Then we have
$d_1(m_{xy},M_{xy}) \leq A \cdot d_2(f(m_{xy}),f(M_{xy})) +AB \leq A \cdot d_2(x,y) +AB.$ 
For the lower bound of $d_1(m_{xy},M_{xy})$, let $M'_{xy}= \max ([x,y] \cap f(\alpha))$ and $m'_{xy}= \min ([x,y] \cap f(\alpha))$.
Then due to the last condition of quasi-isometricity it holds that $d_2(y,M'_{xy}) \leq 2A$ and $d_2(x,m'_{xy}) \leq 2A$. Also from Lemma \ref{L:no-big-crossing} we have 
$$
d_2(f(M_{xy}),M'_{xy}) \leq -C
$$
and
$$
d_2(f(m_{xy}),m'_{xy}) \leq -C.
$$ 
Hence \ 
$
d_2(f(m_{xy}),f(M_{xy})) \geq d_2(x,y) -4A+2C$. \ 
\noindent
In summary, for any $x \leq y \in \beta$ it holds that 
\begin{center}{
$(1/A)\cdot d_2(x,y)-(4A+B-2C)/A$\\
$\leq d_1(m_{xy},M_{xy})\leq A\cdot d_2(x,y)+AB.$}
\end{center}
Set $A' = A$ and $B' = (4A+B-2C)/A +AB $. 
\end{proof}

\begin{cor}\label{C:pre-image}
Let $f:\alpha \rightarrow \beta$ be a quasi-isometry.  There is a $C>0$ so that $|f^{-1}(y)|<C$ for all $y\in \beta$. \qed
\end{cor}

\subsection{Structure theorems}
Elementary properties of the partial order $\Sigma_{QI}^{\omega}$ are in Proposition \ref{Prop:elementary}.
Now we provide several 
structure theorems describing algebraic properties of 
$\Sigma_{QI}^{\omega}$.
\begin{thm}\label{P:infinite-chain}
The partial order  $\Sigma_{QI}^{\omega}$ has a chain $(\alpha_n)_{n \in \mathbb{Z}}$ of the type of integers, that is 
$\forall n\in \mathbb{Z} [\alpha_n <_{QI} \alpha_{n+1}]$. Furthermore, the partial order $\Sigma_{QI}^{\omega}$ has  
a countable anti-chain.
\end{thm}

\begin{proof}
We prove the first part. Take two distinct element in $\Sigma$, say $0$ and $1$, respectively. Let
\[
  \alpha_n = \begin{cases}
    (01)^{2^n}(011)^{2^n}...(01^{2^k})^{2^n}... & (n\geq0) \\
    (01)(01^{2^{2^{-n}}})...(01^{(2^k)^{2^{-n}}})... & (n<0)
  \end{cases}
\]
The idea is the following. The string $\alpha_0$ is of the form:
$$
(01) (011) (01111)\ldots (01^{2^k})\ldots
$$
Let us call the substrings $(01^{2^k})$ blocks of $\alpha_0$. The above definition tells us that $\alpha_1$ is obtained from $\alpha_0$ by doubling each block of $\alpha_0$; $\alpha_{-1}$ is obtained from $\alpha_0$ by removing every other block. This 
is propagated to all $\alpha_n$'s.

We show that $\alpha_n <_{QI} \alpha_{n+1}$ for all $n\geq 0$. For negative $n$, the proof is similar.
To see $\alpha_n \leq_{QI} \alpha_{n+1}$, consider the mapping which sends an interval $(01^{2^{k+1}})$ in $\alpha_n$ to $(01^{2^{k}})^2$ in $\alpha_{n+1}$ in an obvious injective way. This mapping is a quasi-isometry from $\alpha_n$ into $\alpha_{n+1}$.


To see that the converse is not true, assume that $g$ is a quasi-isometry from $\alpha_{n+1}$ into $\alpha_n$. Then using Lemma \ref{L:pre-image} and Lemma \ref{L:no-big-crossing}, one sees that $g$ needs to be strictly monotone on almost every point of $\alpha_{n+1}$ with colour $0$; that is, we need to have $M \in \mathbb{N}$ such that for all
$m,m'>M$ and $m>m'$ we have \ 
$\alpha_{n+1}(m) = \alpha_{n+1}(m') = 0  \Rightarrow g(m) > g(m')$. 
But this is not possible; indeed, let $m_{i,n} \in \omega$ be the position of $i$-th $0$ in $\alpha_n$. If $m_{i,n+1} > M$ and $g(m_{i,n+1}) = m_{j,n}$, then from monotonicity and injectivity we should have $d(g(m_{i,n+1}),g(m_{i+k,n+1})) \geq m_{j+k,n} - m_{j,n}$ for each $k \in \mathbb{N}$. For any $i$ and $j$ we can easily verify
\[
\lim_{k\to\infty} \frac{m_{j+k,n} - m_{j,n}}{m_{i+k,n+1} - m_{i,n+1}} = +\infty,
\]
which means we cannot have any bound $A$ as in Definition 1. This is contradiction, and hence there is no quasi-isometry from $\alpha_{n+1}$ into $\alpha_n$. 



Now we prove the second part. 
Consider the following sequence of strings $\beta_n$, $n\in \omega$:
\[
\beta_n = 010^{2^n}1^{2^n}0^{3^{n}}1^{3^{n}}...0^{k^{n}}1^{k^{n}}...
\]
We claim that this sequence forms an anti-chain in  $\Sigma_{QI}^{\omega}$.  Take any $n,m \in \mathbb{N}$ and suppose $\beta_n \leq_{QI} \beta_m$ via $f$. It suffices to show $n=m$.

Let $A_{n',k'}$ be the ``$k'$-th block'' of zeros in $\beta_{n'}$, i.e. 
\[
A_{n',k'} = [2\sum_{i=0}^{k'-1}i^{n'},2\sum_{i=0}^{k'-1}i^{n'}+(k')^{n'}-1]
\]
Then by easy argument we can show that there exists $k,l \in \mathbb{N}$ such that
(1) $f(A_{n,k}) \subseteq A_{m,l}$, and (2) for every $k' \in \mathbb{N}$,  we have $f(A_{n,k'+k}) \subseteq A_{m,k'+l}$. 
Intuitively, these say that from some point $f$ maps $\beta_n$ in ``block by block'' manner without vacancy. For quasi-isometricity of $f$ we should have an upper and positive lower bound of the rate 
\[
\frac{|A_{n,k'+k}|}{|A_{m,k'+l}|} = \frac{(k'+k)^n}{(k'+l)^m}
\]
with respect to $k'$, since otherwise we do not have any bound $A$ as in Definition 1. Clearly $n=m$ is the only case that satisfies this condition. 
\end{proof}
\noindent
The trivial large scale geometries $[0^{\omega}]$ and $[1^{\omega}]$, as noted above, 
 are minimal elements of $\Sigma_{QI}^{\omega}$. The next theorem shows that there are non-trivial minimal large scale geometries. 
\begin{thm} \label{Thm:minimal}
Let $\{a_n\}_{n \in \mathbb{N}}$ be an unbounded nondecreasing sequence. Then the large scale geometry of the string
$\alpha = 0^{a_0}1^{a_1}0^{a_2}1^{a_3}...0^{a_{2k}}1^{a_{2k+1}}\ldots$
is a minimal element in the partial order $\Sigma_{QI}^{\omega}$.
\end{thm}

\begin{proof}
Let $I_n$ be the interval that corresponds to $p^{a_n}$ ($p \in \{0,1\}$) in $\alpha$, that is, 
\[
I_n = [\sum_{k=0}^{n-1}a_k, (\sum_{k=0}^{n}a_k) -1].
\]
Suppose $\beta \leq_{QI} \alpha$ via $f,A,B$. We claim that $\beta$ is of the form 
$vw'_0w_1w'_1w_2w'_2\ldots w_nw'_n\ldots$, 
where $w_n$ and $w'_n$ are sequences for which we have constants $A',B',D,$ and $n_0$ such that for all $n$ we have $|w'_n|\leq D$, $(1/A') \cdot a_{n_0+n} -B' \leq |w_n| \leq A'a_{n_0+n} +B'$ and $w_n = 0^{|w_n|}$ if $n_0+n$ is even or $w_n = 1^{|w_n|}$ if $n_0+n$ is odd. If the claim is true, then  $\beta =_{QI}\alpha$. 

Let $M_n = \max f^{-1}(I_n)$ and $m_n = \min f^{-1}(I_n)$. From Lemma \ref{L:pre-image}, there exist $A''$ and $B''$ such that
\[
(1/A'') \cdot a_n -B'' \leq M_n - m_n \leq A''a_n +B''
\]
for each $n$. 
Then there are numbers $n_0$ and $D>0$ such that for all $n \geq n_0$ we have $m_{n+2} > M_n$ and $-D \leq m_{n+1} - M_n \leq 1$. 
Indeed, let $D = B-A(C-1)$ and $n_0$ satisfies $a_{n_0}\geq A''(-2AC+3\max\{B,B''\})+2$, where $C$ is a constant such that $f(y) - f(x) \geq C$ for all $x < y$ (Lemma \ref{L:no-big-crossing}). Then $m_{n+1} - M_n \leq 1$ holds since otherwise there exists $x \in \mathbb{N}$ such that $M_n < x < m_{n+1}$ which should be mapped to a point other than $I_n \cup I_{n+1}$, and we can find $x \in [M_n, m_{n+1}-1]$ such that $d(f(x),f(x+1)) > D \geq A+B$, which contradicts the quasi-isometricity of $f$. For a lower bound, if $m_{n+1} < M_n$ then $f(M_n) - f(m_{n+1}) \geq C$ due to Lemma \ref{L:no-big-crossing}, and hence 
\begin{eqnarray*}
m_{n+1} - M_n &=& -d(m_{n+1}, M_n) \\
              &\geq& -A \cdot d(f(m_{n+1}), f(M_n)) -B \\
              &=& A \cdot (f(M_n) - f(m_{n+1})) -B \\
              &\geq& AC -B > -D.
\end{eqnarray*}
We also have
\begin{center}{
$m_{n+2} - M_n \ = $\\
$(m_{n+2} - M_{n+1}) + (M_{n+1} - m_{n+1}) + (m_{n+1} - M_n)$ \\
              $\geq 2AC-2B+(1/A'')\cdot a_{n+1} -B'' >1$.}
\end{center}
Now for each $n \geq n_0$ let $J_n = [M_{n-1}+1, m_{n+1}-1]$ and $J'_n =[m_{n+1},M_n]$ if $m_{n+1}<M_n$, or otherwise $J'_n=\phi$. Let $w_n$ and $w'_n$ be the strings that corresponds to $J_{n_0+n}$ and $J'_{n_0+n}$, respectively. Then $\beta$ is of the form $vw'_0w_1w'_1w_2w'_2\ldots,$ and by the inequality above
\begin{center}{
$|w_n|= m_{n_0+n+1} - M_{n_0+n-1} -1 \geq$\\
$ (1/A'')\cdot a_{n_0+n} +2AC-2B-B''$}
\end{center}
Also, \ 
$|w_n| \leq M_{n_0+n}-m_{n_0+n}+1 \leq A''a_{n_0+n} +B'' +1$.  
Hence letting  $B' = \max\{-2AC+2B+B'',B'' +1\}$ and $A' = A''$ and we have a proof.
\end{proof}

\begin{cor}
The partially ordered set $\Sigma_{QI}^{\omega}$ possesses uncountably many minimal elements.\qed
\end{cor}

Consider the string  \ 
$\alpha=0^{n_0}1^{m_0} 0^{n_1}1^{m_1}\ldots$ from $\{0,1\}^{\omega}$, where $n_i, m_i\geq 1$. 
Call the substrings  $0^{n_i}$ and $1^{m_i}$ the {\bf $0$-blocks and $1$-blocks}, respectively. Define:
\begin{itemize}
\item 
$\mathcal X(0)=\{[\alpha]\mid$ in $\alpha$ all the lengths of $0$-blocks are universally bounded$\}$,
\item 
$\mathcal X(1)=\{[\alpha]\mid$ in $\alpha$  the lengths of all $1$-blocks are universally bounded$\}$,
\item 
$\mathcal X(u)=\{[\alpha]\mid$ in $\alpha$ there is no universal bound on  the lengths of both $0$-blocks and $1$-blocks $\}$,
\item
$\mathcal X(b)=\{[\alpha]\mid$ in $\alpha$ the lengths of both $0$-blocks and $1$-blocks  are universally bounded$\}$.
\end{itemize}
Upward closed sets  are called {\em filterers}, and  downward closed sets  {\em ideals}.
We use this terminology in the next theorem.
\begin{thm} \label{Thm:filters}
The sets $\mathcal X(0)$, $\mathcal X(1)$, $\mathcal X(u)$, $\mathcal X(b)$ satisfy the following:
\begin{enumerate}
\item The sets $\mathcal X(0)$ and $\mathcal X(1)$ are filters. 
\item The set $\mathcal X(u)$ is an ideal. 
\item 
The set $\mathcal X(b)$ is the singleton $\{[(01)^{\omega}]\}$.
\end{enumerate}
\end{thm}

\begin{proof}
We prove that $\mathcal X(0)$ is a filter. Assume that $\alpha\in \mathcal X(0)$ and $\alpha\leq_{QI} \beta$ through a quasi-isometry $f$ with constants $A$ and $B$. We need to show that $\beta\in \mathcal X(0)$. Assume that the length of $0$-blocks in $\beta$ is unbounded. Let $K$ be the length of the largest $0$-block in $\alpha$. Now, take a $0$-block $\beta(i)\ldots \beta(i+n_0)$ in $\beta$, where $n_0$ is sufficiently large. We give bound on  the value of $n_0$ below.  Since $f$ is a quasi-isometry, there exists a
$\beta(i_0)$ that belongs to the block such that $\beta(i_0)$ has a pre-image $\alpha(j_0)$ and $\beta(i_0)$ is within distance $A$ from the center of the $0$-block $\beta(i)\ldots \beta(i+n_0)$. The length of the block in $\alpha$ that contains $\alpha(j_0)$ is bounded by $K$, hence the string 
$$
\alpha(j_0-K-1)\ldots \alpha(j_0)\ldots \alpha(j_0+K+1)
$$ 
must contain $1$. The length of this interval is $2K+2$. There exists a $C$ such that $f$-image of all 
intervals of length $2K+2$ is contained in the intervals in $\beta$ of length at most $C$. So, it must be the case that $n_0\leq C$. 
Hence $\beta \in \mathcal X(0)$. The same proof shows that $\mathcal X(1)$ is a filter.


The proof of Part 2 is similar to the proof above. 


For the last part for all the strings $\alpha$, $\beta$ if in both the lengths of $0$-blocks and $1$-blocks are universally bounded, then  
$\alpha$ and $\beta$ are quasi-isometric to each other since they are colour-equivalent. They are also quasi-isometric to $(01)^{\omega}$. 
\end{proof}

\begin{cor}
For all $\alpha \in \mathcal X(0)$, $\beta \in \mathcal X(1)$,  $\gamma\in \Sigma^{\omega}$,  if $\alpha \leq_{QI} \gamma$ and $\beta\leq_{QI} \gamma$ then $\gamma\sim (01)^{\omega}$.  \qed
\end{cor}

Consider the join operation $\oplus$ that, given $\alpha$ and $\beta$, produces the string 
$\alpha \oplus \beta= \alpha(0)\beta(0)\alpha(1)\beta(1)\ldots$. 
The operation is often used (e.g., in computability and complexity theory) to produce the least upper bounds. 
It turns out this operation is not well-behaved with respect to the large scale geometries.
Indeed, consider the string $\alpha = 010011\ldots(0)^{2^n}(1)^{2^n}\ldots$ and the string $\beta = 101100\ldots (1)^{2^n}(0)^{2^n}\ldots$ \ 
Then $[\alpha \oplus \beta] =[(01)^\omega]$. But, $[(01)^{\omega}] \neq [\alpha]$, $[(01)^{\omega})] \neq\ [\beta]$, and $[\alpha]=[\beta]$.
 

Even though the operation $\oplus$ is not well-behaved with respect to $\sim_{QI}$-classes, the operation can still be useful in constructing counter-examples as shown below. 

\begin{cor}
The sets $\mathcal X(0)$, $\mathcal X(1)$ are not ideals.
\end{cor}

\begin{proof}
Consider the strings $\beta =010011\ldots 0^n 1^n\ldots$ and $\alpha=\beta\oplus 1^{\omega}$.
It is clear that
$\alpha\in \mathcal X(0)$. It is also easy to see that $\beta\leq \alpha$. However, 
$\beta \not \in \mathcal X(0)$. 
\end{proof}

\begin{cor}
Both $\mathcal X(0)$ and $\mathcal X(1)$ have countable chains and anti-chains. In addition, $\mathcal X(u)$ has an infinite  anti-chain consisting of minimal elements.
\end{cor}

\begin{proof}
The chain constructed in Theorem \ref{P:infinite-chain} is in $\mathcal X(0)$. Hence,  both 
$\mathcal X(0)$ and $\mathcal X(1)$ have countable chains.   Let $\beta_n$ be the sequence in the proof of 
Theorem \ref{P:infinite-chain}. Then in the same manner as in the proof of Theorem \ref{P:infinite-chain} we can show 
that $\{\beta_n \oplus 1^\omega \}_{n \in \mathbb{N}} \subset \mathcal{X}(0)$ is an anti-chain. Similarly 
$\{\beta_n \oplus 0^\omega \}_{n \in \mathbb{N}} \subset \mathcal{X}(1)$ is an anti-chain. 
The minimal elements constructed in Theorem \ref{Thm:minimal} are in $\mathcal X(u)$.
\end{proof}



\section{Refining quasi-isometry}

The 
relation $\leq_{QI}$ could be analysed in several ways. One  is to restrict $\leq_{QI}$ to a particular class $K$ 
of infinite strings, and describe the partial order $\leq_{QI}$ restricted to large scale geometers of strings from $K$.
In Section \ref{S:EP},  we fully describe the partial order $\leq_{QI}$ restricted to the class of eventually periodic words. The second way is to refine the definition of $\leq_{QI}$ and study its implications. Section \ref{S:Colour-equivalence} provides
an equivalent, more intuitive, characterisation of quasi-isometry called component wise reducibility. 

\subsection{Eventually periodic spaces} \label{S:EP}

For a string $\alpha$ (that might be finite) consider the set of all colours 
in $\alpha$: $Cl(\alpha)=\{\sigma\in \Sigma\mid \exists i (\alpha(i)=\sigma)\}$. 
Write  $u\sqsubseteq v$ if $Cl(u) \subseteq Cl(v)$. We easily get the following:


\begin{lem}
If $f: \alpha\rightarrow  \beta$ is a quasi-isometry then $Cl(\alpha)\subseteq Cl(\beta)$. So,  if $\alpha\sim_{QI} \beta$ then $Cl(\alpha)=Cl(\beta)$.\qed
\end{lem}
\noindent
A particularly simple strings are eventually periodic: 
\begin{defn}
Metric space $\alpha\in \Sigma^{\omega}$  is   {\bf eventually periodic } if there are finite words $x, u\in \Sigma^\star$ such that $\alpha=xuuu \ldots$.  Call $u$ the {\bf period} of $\alpha$.  
\end{defn}

Let $\alpha=xu^{\omega}$ be eventually periodic word. We 
assume that $Cl(u)\subseteq Cl(x)$ as we can change the prefix $x$ to $xu$. 
With this assumption, we have the following theorem:

\begin{thm}\label{Thm:EPQI1}
For eventually periodic words $\alpha$, $\beta$ we have $\alpha\leq_{QI} \beta$ iff there are \ $x,y, u, v \in \Sigma^{\star}$  such that 
$\alpha=xu^{\omega}$, $\beta=yv^{\omega}$,   $Cl(x)\subseteq Cl(y)$, and $Cl(u) \subseteq Cl(v)$.

\end{thm}

\begin{proof}
If $\alpha \leq_{QI} \beta$ then it is easy to select finite strings $x,y, u, v$ that satisfy the statement of the theorem.

Assume that there are $x,y, u, v \in \Sigma^{\star}$  such that 
\  $\alpha=xu^{\omega}$, $\beta=yv^{\omega}$,   $Cl(x)\subseteq Cl(y)$ and $Cl(u) \subseteq Cl(v)$.   
Construct a quasi-isometry $f$ from $\alpha$ to $\beta$ as follows. First, map the prefix $x$ into prefix $y$ so that colours are preserved. Secondly, we consider the set  $X=\{x_1,\ldots, x_k\}$ of all distinct colours that appear in $u$. 
Now map $x_i$ coloured position in the $k^{th}$ copy of $u$ in the string $\alpha$ into the $x_i$ coloured position in the $k^{th}$ copy of $v$ in the string $\beta$.  Set $A=\max\{|x|, |y|, |u|, |v|\}$ and $B=A$. It is clear that $f$ preserves colours and the quasi-isometry inequality between the distances $d(x,y)$ and $d(f(x), f(y))$ with constants $A, B$ as required. \
\end{proof}

Let $P_1(\Sigma)$ be the set of all non-empty subsets of $\Sigma$. 
Consider the following partial order on the domain:
$$
\mathcal X= \{(A,B) \mid A,B \in P_1(\Sigma) \ \mbox{and} \ A\supseteq B \},
$$
where the partial order on $\mathcal X$ is the component-wise inclusion.
From the theorem above, we obtain the 
a full description of the partial order $\leq_{QI}$ restricted to large scale geometers of eventually periodic strings.

\begin{cor}\label{Cor:EP-order}
The partial order $\leq_{QI}$ restricted to the set 
$EP=\{[\alpha] \mid \alpha $ is eventually periodic string over $\Sigma\}$ 
is isomorphic to the partial order $\mathcal X$.  \qed
\end{cor}



Let $f:\alpha\rightarrow \beta$ be a quasi-isometric map, and let $P$ be a property. Say that $f$ is {\em eventually $P$}
if there is an $i$ such that $f$ restricted to the suffix \ $\alpha_i=\alpha(i) \alpha(i+1) \alpha(i+2)\ldots$ \ 
satisfies $P$. A quasi-isometry $f$ does not need to be eventually order preserving map; neither $f$ needs to be eventual embedding (that is, eventually injective map). Now we show that if $\alpha$ and $\beta$ are 
eventually  periodic words, then $f$ can be computable eventually order preserving injective map.


\begin{thm}\label{Thm:EPQI2}
For eventually 
periodic  strings $\alpha$ and $\beta$ such that $\alpha\leq_{QI} \beta$ 
there exists a computable quasi-isometric map $f_{\alpha}:\alpha \rightarrow \beta$ which is   eventually order preserving
and  injective. 
\end{thm}

\begin{proof}
Consider $\alpha=xuuu\ldots$ and $ \beta=yvvv\ldots$. 
We assume that $Cl(x)\subseteq Cl(y)$ and $Cl(u) \subseteq Cl(v)$.  
Our desired mapping $f_{\alpha}$ maps $x$ into $y$ by preserving colours. (That is where $f_{\alpha}$ need not be order preserving embedding). So, it suffices to 
construct a computable
quasi-isometric embedding from  $\alpha_1=u^{\omega}$ 
to $\beta_1=v^{\omega}$.
Let $X=\{x_1,\ldots, x_k\}$ be the set of all colours that appear in $\alpha_1$, and hence in $\beta_1$.  
Let $a_i$ be the number of times that colour $x_i$ appears in $v$. Clearly $a_i\geq 1$. Consider the new string $v_1$ obtained by writing $v$
exactly $a_1+\ldots +a_n$ times, that is, 
$v_1=(v)^{a_1+a_2+\ldots + a_n}$.
There exists an embedding  $f'$ from $u$ into $v_1$ that preserves colours and the order. Just like in the theorem above, we can propagate this mapping $f'$ to a colour preserving  embedding $f_{\alpha}$ from $\alpha_1$ into $\beta_1$.
The map $f_{\alpha}:\alpha \rightarrow \beta$ thus built is a  desired function. 
\end{proof}

\subsection{Componentwise Reducibility} \label{S:Colour-equivalence}

We formulate an equivalent more intuitive condition for quasi-isometry between coloured metric spaces.


\begin{defn}
Say $\alpha$ is {\bf componentwise reducible} to $\beta$, written  $\alpha \leq_{CR} \beta$,
if we can write 
$\alpha=u_1u_2\ldots$ and $\beta=v_1v_2\ldots$ such that
$u_i\sqsubseteq v_i$ for each $i$ and $|u_i|,|v_i|$ are uniformly bounded by a constant $C$. Call these presentations of $\alpha$ and  $\beta$ {\bf witnessing partitions}.
\end{defn}

It is clear  that $\alpha \leq_{CR} \beta$ implies $\alpha \leq_{QI} \beta$: any colour-preserving map that maps each interval $u_i$ in $\alpha$ to $v_i$ in $\beta$ is quasi-isometry. 
Showing the converse is not trivial. The main difficulty is showing that $\leq_{CR}$ is transitive. 


\begin{defn}
An {\bf atomic crossing map} is a function $f:\alpha \to \beta$ of the following form: we have $\{a_i,b_i\}_{i \in I}$, where $I$ is an at most countable index set such that $\exists C \forall i [(a_{i} < b_{i} < a_{i+1}) \land (b_i - a_i \leq C)]$ and

\[
f(a) = \left\{ \begin{array}{lll}
    b_i & (a = a_i) \\
    a_i & (a = b_i) \\
    a  & (otherwise)
  \end{array} \right.
\]
\end{defn}
\noindent
Clearly, every atomic crossing map is bijective. The next result is a step towards transitivity of $\leq_{CR}$. 

\begin{prop}\label{Prop:decomposition}
Any quasi-isometry $f:\alpha \to \beta$ can be decomposed into the following form:
\[
\alpha \xrightarrow{f_1} \gamma_1 \xrightarrow{f_2} \gamma_2 \xrightarrow{f_3} \beta,
\]
where $f_1$ is a monotonic injection, $f_2$ is a monotonic surjection, $f_3$ is a bijection, and $f_1,f_2,f_3$ are all quasi-isometric. 
Furthermore, $f_3$ can be decomposed into the following form for some $n\geq1$:
\[
\gamma_2 \xrightarrow{g_1} \delta_1 \xrightarrow{g_2} \ldots \xrightarrow{g_n} \beta,
\]
where each $g_k$ is an atomic quasi-isometry. 
\end{prop}

\begin{proof}
Let $f:\alpha\to\beta$ be a quasi-isometry. First we decompose it to $\alpha \xrightarrow{f'} \gamma_2 \xrightarrow{f_3} \beta$, where $f'$ is monotonic and $f_3$ is a bijection. By Small Cross Over Lemma we have a bound $D < 0$ such that $f(m)-f(n)\geq D$ for all $n<m$. Let $\{(n_i, m_i)\}_{i\in I}$ be the set of all pairs of natural numbers such that $n_i< m_i$ and $f(m_i)-f(n_i)= D$, where $I$ is an at most countable index set that depends on $f$. We define $\tilde \beta$ as the sequence that is same as $\beta$ except that colours of the position $n_i$ and $m_i$ are swapped for each $i \in I$: that is,

\[
  \tilde \beta(n) = \begin{cases}
    \beta(m_i) & (n=n_i) \\
    \beta(n_i) & (n=m_i) \\
    \beta(n) & (otherwise)
  \end{cases}
\]
Notice that $n_i,m_i,n_j$ and $m_j$ are all distinct element for distinct $i,j \in I$, and hence the above definition is well-defined. Also let $\tilde f$ be a function that executes this swapping, i.e. $\tilde f(m_i)=n_i$, $\tilde f(n_i)=m_i$ and $\tilde f(n)=n$ otherwise. It is easy to show that $\tilde f \circ f$ is a QI map from $\alpha$ to $\tilde \beta$ such that $\tilde f \circ f(m)-\tilde f \circ f(n)\geq D+1$ for all $n<m$. Iterating this procedure $(-D)$ times we get $\gamma_2$ and a bijective QI map $g:\beta \to \gamma_2$ such that $g\circ f$ is monotonic. let $f' =g\circ f$ and $f_3 = g^{-1}$. From the construction it is clear that $g^{-1}$ is a finite composition of certain atomic crossing functions.

Now we decompose $f'$ to $\alpha \xrightarrow{f_1} \gamma_1 \xrightarrow{f_2} \gamma_2$, where $f_1$ is a monotonic injection and $f_2$ is a monotonic surjection. For $\gamma_2 = c_1c_2\ldots$, let $\gamma_1 = (c_1)^{n_1}(c_2)^{n_2}\ldots$, where $n_i = \max\{1, |f_1^{-1}(c_i)|\}$. Let $f_1$ be a map that sends each interval $f'(c_i)$ in $\alpha$ to $(c_i)^{n_i}$ in $\gamma_1$ in an obvious injective, monotonic way. Also let $f_2$ be a map that sends each $(c_i)^{n_i}$ in $\gamma_1$ to $c_i$ in $\gamma_2$. This is the desired
decomposition. 
\end{proof}

\begin{cor}
If $\alpha \leq_{QI} \beta$, then there is $\beta' \sim_{QI} \beta$ so that $\alpha \leq_{QI} \beta'$ via strictly monotonic 
quasi-isometry. \qed
\end{cor}

\begin{proof}
The strings $\gamma_1$ is a desired $\beta'$ as $\beta'$ is obtained through 
the composition of quasi-isometric maps $f_3^{-1}$ and $\lambda n.\min f_2^{-1}(n)$ applied to $\beta$.
\end{proof}

Informally, the proposition above says that any quasi-isometry can be decomposed into three parts. The first part is 
the mapping $f_1$ which can be viewed as a ``pollution" of $\alpha$, the second part is the map $f_2$  that can be called a ``collapsing'' map, and  the third part is $f_3$ that can be called ``mixing'' since it mixes atomic crossing maps. Note that the mixing part is the one that makes things complicated: below we show that mixing functions preserve componentwise reducibility.

\begin{lem}\label{Lem:acf}
Suppose $\alpha \leq_{CR} \beta$ and $\beta \leq_{QI} \gamma$ via an atomic crossing map $f: \beta \to \gamma$. Then $\alpha \leq_{CR} \gamma$.
\qed
\end{lem}

\begin{proof}
Let $\alpha = u_1u_2\ldots$ and $\beta = v_1v_2\ldots$ be witness partitions. Let $\gamma = w_1w_2\ldots$ be a partition such that $|w_n| = |v_n|$ for all $n$. We construct a strictly increasing sequence $\{n_k\}_{k \in \mathbb{N}}$ such that $v_{n_k}\ldots v_{n_{k+1}-1}\sqsubseteq w_{n_k}\ldots w_{n_{k+1}-1}$ for each $k$. Through the proof we identify a finite sequence $u_n$ in a partition $\alpha = u_1u_2\ldots$ and an interval $[|u_1\ldots u_{n-1}|,|u_1\ldots u_n|-1]$ in $\alpha$, and write as ``$a \in u_n$'' for given $a \in \mathbb{N}$.

We can assume $f$ satisfies the following conditions without loss of generality:

\begin{enumerate}
\item $f$ is characterised by $\{a_i,b_i\}_{i \in \mathbb{N}}$ such that $C(a_i) \neq C(b_i)$ for all $i \in \mathbb{N}$.
\item For each $i$, $a_i \in u_i$ and $b_i \in u_{i+1}$.
\end{enumerate}
The following claim is the crucial part of our proof.

\begin{clm}
Let $n\in \mathbb{N}$ be given. If $C(b_n) \in C(w_{n+i})$ for some $1 \leq i\leq |\Sigma|$, then there exists $j$ such that $1 \leq j \leq |\Sigma|(\Sigma +1)$ such that $v_{n+1} \ldots v_{n+j} \sqsubseteq w_{n+1} \ldots w_{n+j}$ and $c(b_{n+j}) \in c(w_{n+j+i'})$ for some $1 \leq i'\leq |\Sigma|$.
\end{clm}

Note that for given $n$, there are $n_1,n_2 \in \{n, \ldots, n+|\Sigma|\}$ such that $n_1 < n_2$ and $C(b_{n_1}) = C(b_{n_2})$: this is immediate from the pigeonhole principle. This implies that $C(b_{n_1}) \in C(w_{n_2})$ and $n_2-n_1 \leq |\Sigma|-1$.
Also for given $n$, let 
$$A = \{m \ | \ m > n \land C(a_m) \not\in C(w_{n+1}\ldots w_m)\}.$$
Then $|A| \leq |\Sigma|-1$, as $m \in A$ only if $C(a_m) \neq C(a_n)$, and for any distinct $m, m' \in A$ we have $C(a_m)\neq C(a_m')$. 

Now assume $C(b_n) \in C(w_{n+i})$ holds for some $1 \leq i\leq |\Sigma|$. From the facts above, we have $1 \leq m \leq |\Sigma|^2$ such that for all $m' \leq |\Sigma|$ we have $C(b_n), C(a_{n+m+m'}) \in C(w_{n+1}\ldots w_{n+m+m'})$. Notice that from this we have $v_{n+1}\ldots v_{n+m+m'} \sqsubseteq w_{n+1}\ldots w_{n+m+m'}$. Finally, we can find $m'$ and $i'$ such that $m' < m'+i' \leq |\Sigma|$ and $C(b_{n+m+m'}) = C(b_{n+m+m'+i'})$, hence $C(b_{n+m+m'}) \in C(w_{n+m+m'+i'})$, where $1 \leq i' \leq |\Sigma|$. This $m+m'$ is the desired $j$.
\smallskip

Find $n$ so that $u_1 \ldots u_n \sqsubseteq w_1 \ldots w_n$ and $C(b_n) \in C(w_{n+i})$ for some $1 \leq i\leq |\Sigma|$, then iterate finding $j$ in the claim above. This is a procedure that gives the desired $\{n_k\}$. 
\end{proof}

\begin{thm}
$\alpha \leq_{QI} \beta$ implies $\alpha \leq_{CR} \beta$.
\end{thm}

\begin{proof}
Let $f:\alpha \to \beta$ be a quasi-isometry. Decompose it into $\alpha \xrightarrow{f_1} \gamma_1 \xrightarrow{f_2} \gamma_2 \xrightarrow{f_3} \beta$ according to Proposition \ref{Prop:decomposition}. For any monotonic QI map $g:\alpha'\to\beta'$ we easily have $\alpha'\leq_{CR}\beta'$: Indeed, enumerate all elements of an image of $g$ as $\{b_n\}_{n\in\mathbb{N}}$, where $b_n < b_{n+1}$ for all $n$, let $u_n = g^{-1}(b_n)$ and $v_n = [b_n, b_{n+1}-1]$. Then $\alpha'= u_1u_2\ldots$ and $\beta'=v_1v_2\ldots$ are witnessing partitions. Hence we have $\alpha \leq_{CR}\gamma_2$, and Lemma \ref{Lem:acf} completes the proof.
\end{proof}

One could attempt to go further, namely say $\alpha$ and $\beta$ are {\it colour-equivalent} if they are  componentwise reducible to each other via the {\it same} witnessing partitions. It is clear that if $\alpha$ and $\beta$ are colour-equivalent, then they are quasi-isometric.  Can we show that $\alpha \sim_{QI} \beta$ implies their colour-equivalence? 

\begin{prop} \label{Prop:Colour1}
There are colour-equivalent $\alpha$ and $\beta$ such that no eventual embeddings witness quasi-isometry between $\alpha$ and $\beta$.
\end{prop}
\begin{proof}
Consider the following coloured metric spaces over the alphabet $\{0,1\}$:
$$
\alpha=001^2001^4\ldots 1^{2n}00 \ldots  \hspace{1mm} \mbox{and} \hspace{1mm} \beta=0101^2\ldots 01^n\ldots.
$$

We now prove that no eventual embeddings exist witnessing quasi-isometry between $\alpha$ and $\beta$.


Assume that  $f:\alpha\rightarrow \beta$ and $g:\beta\rightarrow \alpha$ are eventual embeddings with witness constants $A_f, B_f$ 
and $A_g, B_g$. Let $i_n$ and $i_n+1$ be the sequence of consecutive  positions in $\alpha$ with colour $0$.
Since $f$ is quasi-isometry we have 
$$
|f(i_n+1)-f(i_n)|\leq A_f |i_{n+1}-i_n| +B_f=A_f+B_f.
$$ 
There exists a position $i_n$ such that $f ({i_n})< f(i_n+1)$ and the number of consecutive $1$s immediately  on the right side of $f(i_n)$ is greater that $A_f+B_f$. Hence, $f(i_n+1)-f(i_n)> A_f+B_f$. This is a contradiction with quasi-isometry.    
\end{proof}

\begin{prop}\label{Prop:Colour2}
There are sequences $\alpha$ and $\beta$ such that $\alpha$ and $\beta$ are quasi-isometric via monotonic embeddings but
$\alpha$ and $\beta$ are not colour-equivalent.
\end{prop}

\begin{proof}
Let $\Sigma = \{0,1,*\}$ and define $\alpha,\beta \in \Sigma^\omega$  as follows.
We first define $\beta$ by:
$$
 0 1 0 * \  \  0 * 1 * 0 * * * \  \  \ldots \  0 (*)^n 1 (*)^n 0(*)^{2n+1} \ \ldots
$$
Thus, the $n^{th}$-part of $\beta$ is the string  $0 (*)^n 1 (*)^n 0(*)^{2n+1}$. We re-write $\beta$ with subscripts on each $0$ and $1$ as pointers for easier  readability:
$$
0_0 1_0 0_{\frac{1}{2}} * 0_1 * 1_1 * 0_{\frac{3}{2}} * * * ... 0_n (*)^n 1_n (*)^n 0_{\frac{2n+1}{2}} (*)^{2n+1}...
$$
Let $\alpha$ be obtained from $\beta$ by omitting the second occurrence of $0$ in the $n^{th}$-substring of $\beta$. So, 
$$
\alpha=0 1 * \  \  0 * 1 * * * *  \ \  \ldots \  0 (*)^n 1 (*)^n (*)^{2n+1} \ \ldots
$$
So, the $n^{th}$-substring of $\alpha$ is $0 (*)^n 1 (*)^n(*)^{2n+1}$. As above we re-write $\alpha$ with with subscripts
as pointers:
$$
\alpha =0_0 1_0 * 0_1 * 1_1 * * * * ... 0_n (*)^n 1_n (*)^{3n+1}...
$$

Constructing a quasi-isometry  $f: \alpha \to \beta$ is straightforward, since eliminating $0_{\frac{2n+1}{2}}$ for each $n \in \mathbb{N}$ from $\beta$ we have the sequence $\alpha$. We define $g: \beta \to \alpha$ as a map that sends certain intervals in $\beta$ to intervals in $\alpha$ monotonically and as equally as possible, as follows:
\begin{itemize}
\item $g$ maps each sequence $0_n(*)^n1_n$ in $\beta$ to the sequence $0_{2n}(*)^{2n}1_{2n}$ in $\alpha$;
\item $1_n (*)^n 0_{\frac{2n+1}{2}}$ in $\beta$ to $1_{2n} (*)^{6n+2}0_{2n+1}$ in $\alpha$; and
\item $0_{\frac{2n+1}{2}} (*)^{2n+1}0_{n+1}$ in the string $\beta$ to the sequence $0_{2n+1}(*)^{2n+1} 1_{2n+1} (*)^{6n+3}0_{2(n+1)}$ in $\alpha$.
\end{itemize}
Then $d(i,j) \leq d(g(i),g(j)) \leq 6d(i,j)$ and $g$ is a monotonic embedding, thus it is a witness for quasi-isometric embedding.

For non-colour equivalence, assume the contrary and let $\alpha = u_1u_2...$ and $\beta = v_1v_2...$ be a witness. We should have only finitely many $u_i$ and $v_i$ that contain more than one occurrences of 0 and/or 1, since otherwise we do not have a bound of $|u_i|$ or $|v_i|$. Let $I \in \mathbb{N}$ be a number such that for all $i \geq I$ this does not happen.

Let $p\geq I$ and suppose $u_p$ contains $1_n$. Then $v_p$ contains $1_{n+k}$ for some $k \in \mathbb{Z}$, and for all $1_{n'}$ and $p'\geq I$ such that $u_{p'}$ contains $1_{n'}$, $v_{p'}$ contains $1_{n'+k}$ (otherwise some $u_{p'}$ or $v_{p'}$ contains multiple 1s). 

Now suppose $p' > p \geq I$ and $u_p$ and $u_{p'}$ contain $1_n$ and $1_{n+1}$, respectively. 
Then from the construction of $\alpha$ there should be exactly one $q \in [p+1, p'-1]$ with the corresponding interval $u_q$ contains 0.
On the other hand $v_p$ and $v_{p'}$ contain $1_{n+k}$ and $1_{n+k+1}$, respectively, and from the construction of $\beta$ we should have two elements in $[p+1, p'-1]$ with the corresponding interval contains 0, which is the contradiction.

\end{proof}





\section{B\"uchi automata and large scale geometries} \label{S:Buchi}

Let $L \subseteq \Sigma^{\omega}$ be a language, where  we assume that $\Sigma=\{0,1\}$.  The notion of quasi-isometry leads us to consider the quasi-isometry types of strings from $L$: 

\begin{defn}
An {\bf atlas} is a set of quasi-isometry types. In particular, the atlas defined by the language $L$ is
the set \ $[L]=\{[\alpha] \mid \alpha \in L\}$, where $[\alpha]$ is the quasi-isometry type of $\alpha$.
\end{defn}

A question of a general character is concerned with understanding the set $[L]$ given a description of $L$. In particular, a natural question is if one can decide that $[L_1]=[L_2]$ given descriptions of languages $L_1$ and $L_2$. In this section, we study
the atlases defined by B\"uchi automata recognisable languages. We show that for B\"uchi automata recognisable languages $L_1$ and $L_2$ there is an efficient algorithm that decides if $[L_1]=[L_2]$.  We recall basic definitions for B\"uchi automata.

\begin{defn}  \label{df: buchi automata}
A {\bf B\"uchi automaton} $\mathcal M$  is  a quadruple $(S,\iota,\Delta,F)$, where
$S$ is a finite  set of {\it states}, $\iota \in S$ is the {\it initial state},  
$\Delta \subset S \times \Sigma  \times S$ is the
{\it transition table}, and $F \subseteq S$ is the set of {\it accepting states}.
\end{defn}

A {\it run} of $\mathcal M$ on $\alpha=\sigma_0 \sigma_1 \dots \in \Sigma^{\omega}$ is a sequence of states 
\ 
$r=s_0,s_1,\dots $ \ such that $s_0 = \iota$ and $(s_i,\sigma_{i},s_{i+1}) \in \Delta$ for all $i \in
\omega$. The run is {\em accepting } if the set $Inf(r)=\{s: \exists^\infty i (q_i=s)\}$ has a  
state from~$F$. The automaton $\mathcal M$ {\em accepts} the string $\alpha$ if it has an accepting run on it.
The {\it language} accepted by $\mathcal M$, denoted $L(\mathcal M)$, 
is the set of all infinite words  accepted by $\mathcal M$. 


Let $\mathcal M$ be a B\"uchi automaton. Our goal is to describe the atlas defined by the language $L(\mathcal M)$. We call such $[L]$ {\em B\"uchi recognisable}. We need to do some state space analysis of the automation $\mathcal M$.




A {\em loop} (in $\mathcal M$) is a path $L$ of states $s_1, \ldots, s_{n+1}$ in the state space $S$ such that $s_1=s_{n+1}$.
We say that a word $\sigma_1\ldots \sigma_n$ {\em labels} the loop $L$ if $(s_i, \sigma, s_{i+1})\in \Delta$ for $1\leq i \leq n$. 
If all symbols $\sigma_i$ are $0$ only (or $1$ only), then we say that the loop is {\em 0-loop} (or respectively, {\em 
1-loop}). 
In the loop $L$ we do not require that the states are pairwise distinct. Any state in a loop  is a {\em loop state}.


The state space $S$ can naturally be considered as a directed graph, where the edges between states are transitions of $\mathcal M$ with labels removed. Hence, we can write $S$ as a disjoint union of its strongly connected components (scc{\em s}). 
Call a scc $X\subseteq S$ {\em non-trivial} if $|X|>1$, where $|X|$ is the cardinality of $X$. Thus we write $S$ as the union
$T\cup S_1\cup \ldots \cup S_k$, where $T$ is the union of trivial strongly connected components and each $S_i$ is a nontrivial scc. Every successful run of $\mathcal M$ on every infinite string $\alpha$ eventually ends up in one of the scc{\em s} $S_i$.
Consider B\"uchi automata $\mathcal M_1$, $\ldots$, $\mathcal M_k$ such that the initial states and the state diagrams of $\mathcal M_i$ all  coincide with that of  $\mathcal M$ but $F_i=F\cap S_i$. It is clear that 
$$
(\star) \hspace{1cm} L(\mathcal M)=L(\mathcal M_1) \cup \ldots \cup L(\mathcal M_k).
$$
Hence describing the atlas $[L(\mathcal M)]$  boils  down to describing the atlases of $[L(\mathcal M_i)]$, $i=1, \ldots, k$.  Assume that $\alpha$ is accepted by $\mathcal M_i$. Then we can write $\alpha$ as $v\alpha'$ such that an accepting run of $\mathcal M_i$ after reading $v$ stays inside $S_i$. Therefore, the large scale geometry $[\alpha]$ of $\alpha$ is quasi-isometric to either $[0\alpha']$ or $[1\alpha']$ (or both). Therefore, we can postulate the following assumption till the end of this section unless otherwise stated.

\smallskip
\noindent
{\bf Postulate}. The state space $S$ can be written as $\{q_0\}\cup S'$ such that (1) $q_0$ is the only initial state and $q_0\not \in S'$, (2) $S'$ is the only nontrivial strongly connected component of $\mathcal M$, and (3) any transition from $q_0$ goes into $S'$. 

\begin{lem} \label{L:Sigma-omega} If $\mathcal M$ has a 0-loop and a 1-loop,  then the atlas $[L(\mathcal M)]$ coincides with the atlas  $[\Sigma^{\omega}]\setminus X$ for some $X\subseteq  \{[0^\omega], [1^{\omega}], [10^\omega], [01^{\omega}]\}$. 
\end{lem}

\begin{proof}
Assume that $L_0$ is a 0-loop and $L_1$ is a 1-loop. Let $\alpha$ be a string with $\omega$-many $0$s and
$\omega$-many $1$s. Say, $\alpha=0^{n_0}1^{m_0}0^{n_1}1^{m_1}\ldots$. We can build a  $\beta$ of the form
$$
u \ \ 0^{n_1'} v 1^{m_1'}w  \ \ 0^{n_2'} v1^{m_2'} w \  \  0^{n_3'} v1^{m_3'} w  \ \ \ldots \  \ 0^{n_i'} v1^{m_i'} w \  \ \ldots
$$
such that (1) $\beta$ is accepted by $\mathcal M$, (2) $n_i=n_i'+c_i$ and $m_i=m_i'+c_i'$, where $c_i<|S|$ and $c_i' <|S|$ for all $i$, (3) $v$ is a string that labels a path from a state $s_0$ in $L_0$ to a state $s_1$ in $L_1$, and (4) $w$ is a string that labels a path from $s_1$ to $s_0$, and (5) $v$ is a string that labels a path from the initial state to $s_0$.  This implies that $\alpha$ and $\beta$ are colour-equivalent. Hence they are quasi-isometric. 


Note that none of the strings $0^\omega$, $1^{\omega}$, $10^\omega$, $01^{\omega}$ is quasi-isometric to a string with both $\omega$-many $0$s and $1$s. Also, these four strings are pairwise not quasi-isometric. Hence, the choice of $X$ is dependent on whether or not $\mathcal M$ accepts some (or all) of these four strings. For instance, $0^{\omega}\in X$ iff $0^{\omega}$ is not accepted by $\mathcal M$. 
\end{proof}

\begin{lem}\label{L:no-0-1}
If $\mathcal M$ has no 0-loop and has no 1-loop, then the atlas $[L(\mathcal M)]$ equals the singleton atlas 
$[\{(01)^\omega\}]$.
\end{lem}

\begin{proof} Let $\alpha=\sigma_0 \sigma_1\ldots $ be a string accepted by $\mathcal M$.   Let $\rho=s_1, s_2, \ldots$ be an accepting run of $\mathcal M$ on $\alpha$. We can write $\alpha$ as $w_1 w_2\ldots$ so that (1) the length of each $w_i$ is bounded by $|S|$ and (2) the run $\rho$ along each $w_i$ has a loop; that loop within $w_i$ contains both $0$ and $1$. This implies that $(01)^{\omega}$ and $\alpha$ are colour-equivalent. Hence they are quasi-isometric. 
\end{proof}

\begin{lem}\label{L:no-1}
Suppose that $\mathcal M$ contains no 1-loop but has a 0-loop $L_0$. Then the atlas $[L(\mathcal M)]$ coincides with one of the following atlases: 
\begin{center}{$[\{0^{\omega}\}]$, $[\{10^{\omega}\}]$,  $[\{0^{\omega}, 10^{\omega}\}]$, $[\{0^{n_1} 10^{n_2}10^{n_3}\ldots \mid n_i>0\}
\setminus  \{0^{\omega}\}]$ or $[\{0^{n_1} 10^{n_2}10^{n_3}\ldots \mid n_i>0\}\cup \{0^{\omega}\}]$.}
\end{center}
\end{lem}

\begin{proof}
The assumption of the lemma implies that if $\alpha$ is accepted by $\mathcal M$ then (1) $\alpha$ contains infinitely many $0$s and (2) the lengths of any sequence of consecutive $1$s occurring in $\alpha$ is bounded by $|S|$. If  $L(\mathcal M)=V0^{\omega}$ for some regular language $V\subseteq \Sigma^\star$, then the atlas $[L(\mathcal M)]$ is either
$[\{0^{\omega}\}]$ or $[\{10^{\omega}\}]$ or $[\{0^{\omega}, 10^{\omega}\}]$. \ So, we assume that $\mathcal M$ accepts a string with $\omega$ many $1$s'.  
Let $\alpha$ be any string of the form $0^{t_1} 10^{t_2}1\ldots$, where all $t_i>0$ for all $i$. We can build a string $\beta$ of the form 
$u \ \ 0^{n_1} v  \ \ 0^{n_2} v \  \  0^{n_3} v  \ \ \ldots \  \ 0^{n_i} v \  \ \ldots$
such that (1) $\beta$ is accepted by $\mathcal M$, (2) $t_i=n_i+c_i$, where $c_i<|S|$ for all $i$, and 
(3) $v$ is a string that labels a path from a state $s_0$ in $L_0$ back to $s_0$, and $v$ contains $1$.
The string $\beta$ is quasi-isometric to string $\beta'$ obtained from $\beta$ by replacing all occurrences of $v$ with $1$ and removing the prefix $u$. In turn $\beta'$ is colour equivalent to $\alpha$. Hence, $\alpha\sim_{QI} \beta$.


Now suppose that $\beta$ is accepted by $\mathcal M$ and $\beta$ has infinitely many $1$s. Then $\beta$ is of the form
$u0^{n_1}1^{m_1} 0^{n_2} 1^{m_2}\ldots$ where we have $n_i, m_i>0$ for all $i$. As noted above the numbers $m_i$ are uniformly bounded. The string $\beta$ is quasi-isometric to string $\beta'$ obtained from $\beta$ by replacing all $1^{m_i}$ 
with just $1$ and removing the prefix $u$. Clearly $\beta'$ is of the form desired. Note that $\mathcal M$ accepts the string quasi-isometric to $10^{\omega}$. Hence, if $0^{\omega}$ is accepted by $\mathcal M$ then the atlas $[L(\mathcal M)]$ coincides with $[\{0^n_1 10^{n_2}10^{n_3}\ldots \mid n_i>0\}\cup \{0^{\omega}\}]$. Otherwise,  the atlas $[L(\mathcal M)]$ equals
$[\{0^n_1 10^{n_2}10^{n_3}\ldots \mid n_i>0\} 
\setminus  \{0^{\omega}\}]$. 
\end{proof}
\begin{lem}\label{L:no-0}
Suppose that $\mathcal M$ contains no 0-loop but has a 1-loop $L_0$. Then the atlas $[L(\mathcal M)]$ coincides with one of the following atlases: 
\begin{center}{$[\{1^{\omega}\}]$, $[\{01^{\omega}\}]$, $[\{1^{\omega}, 01^{\omega}\}]$, $[\{1^{n_1} 01^{n_2}01^{n_3}\ldots \mid n_i>0\} 
\setminus  \{1^{\omega}\}]$ or\\
\  \   \   \  \  \  \   \   \  \  \ $[\{1^{n_1} 01^{n_2}01^{n_3}\ldots \mid n_i>0\}\cup \{1^{\omega}\}]$.\qed}
\end{center}
\end{lem}

\noindent
We now remove our postulate put on the state space of B\"uchi automata. Using equality $(\star)$ given above, and the lemmas, we obtain the following characterisation result:

\begin{thm} \label{Thm:Atlas}
Any B\"uchi recognisable atlas $[L]$ 
is a union from the following list of atlases:
\begin{itemize}
\item $[\Sigma^{\omega}]\setminus X$, where $X\subseteq  \{[0^\omega], [1^{\omega}], [10^\omega], [01^{\omega}]\}$.
\vspace{-0.5mm}
\item  $[\{0^{n_1} 10^{n_2}10^{n_3}\ldots \mid n_i>0\} 
\setminus  \{0^{\omega}\}]$, \ $[\{1^{\omega}\}]$,
\vspace{-0.5mm}
\item $[\{0^{n_1} 10^{n_2}10^{n_3}\ldots \mid n_i>0\}\cup \{0^{\omega}\}]$,  \ $[\{0^{\omega}\}]$,
\vspace{-0.5mm}
\item     $[\{01^{\omega}\}]$,  $[\{10^{\omega}\}]$,  $[\{0^{\omega}, 10^{\omega}\}]$, $[\{1^{\omega}, 01^{\omega}\}]$,
\item $[\{1^{n_1} 01^{n_2}01^{n_3}\ldots \mid n_i>0\}
\setminus  \{1^{\omega}\}]$,
\vspace{-0.5mm}
\item $[\{1^{n_1} 01^{n_2}01^{n_3}\ldots \mid n_i>0\}\cup \{1^{\omega}\}]$. \qed
\end{itemize}
\end{thm}

\noindent
We obtain the following complexity-theoretic result solving the equality problem for B\"uchi recognisable 
atlases.  
In comparison, the equality problem for B\"uchi automata is  PSPACE complete.

\begin{cor}\label{Cor:Atlas-equality}
There exists an algorithm that, given  B\"uchi automata $\mathcal A$ and $\mathcal B$, decides if the atlases 
$[L(\mathcal A)]$ and $[L(\mathcal B)]$ coincide. Furthermore, the algorithm runs in linear time on the size of the input automata.
\end{cor}

\begin{proof}
Let $\mathcal M$ be a B\"uchi automaton. Our goal is to find atlases from the list provided in Theorem \ref{Thm:Atlas} such that the union of these atlases coincides with the atlas $[L(\mathcal M)]$. For this, we decompose $\mathcal M$ into $k$ automata $\mathcal M_1$, $\ldots$, $\mathcal M_k$ so that the equality $(\star)$ holds and each $\mathcal M_i$ satisfies the postulate. This decomposition can be done in linear time on the size of $\mathcal M$ (e.g. using Tarjan's algorithm that decomposes a directed graph into strongly connected components). Now for each $\mathcal M_i$, we check the assumptions of Lemma \ref{L:Sigma-omega} through Lemma \ref{L:no-0}. This can also be done in linear time.  


As an example, the assumptions of Lemma \ref{L:Sigma-omega} can be checked as follows. Assume that the automaton under consideration is $\mathcal M_i$. In order to check if $\mathcal M_i$ contains a 0-loop, we  remove all transitions labeled with $1$ from the state diagram of $\mathcal M_i$. This gives us a directed graph whose edges are transitions labelled with $0$. In this graph we check if there is a loop. If there  exists a loop, then $\mathcal M_i$ has a 0-loop. Otherwise $\mathcal M_i$ does not have a 0-loop. Similarly, to check if $\mathcal M_i$ contains a 1-loop, we  remove all transitions labeled with $1$. This gives us a directed graph. If there is a loop in this directed graph,  then $\mathcal M_i$ has a 1-loop;  Otherwise, not. To find the set $X$ from (the statement of) the Lemma we just need to check which of the strings $0^{\omega}$, $1^{\omega}$, $10^{\omega}$, $01^{\omega}$ is accepted by $\mathcal M_i$. This can also be done in linear time on the size of $\mathcal M_i$. 


The argument shows that we can find, in linear time on size of $\mathcal M$, the  atlases from the list provided in Theorem \ref{Thm:Atlas} such that the union of these atlases is the atlas $[L(\mathcal M)]$. Thus, given  B\"uchi automata $\mathcal A$ and $\mathcal B$ we can decide in linear time if  
$[L(\mathcal A)]=[L(\mathcal B)]$. 
\end{proof}
\vspace{-4mm}

\section{The quasi-isometry problem}\label{S:complexity}

The quasi-isometry problem $QIP$ is to find if, given strings $\alpha$ and $\beta$, there is a quasi-isometry from $\alpha$ to $\beta$: 
$$
QIP=\{(\alpha, \beta) \mid \alpha, \beta \in \Sigma^{\omega} \ \& \ [\alpha]\leq_{QI} [\beta]\}.
$$
Let $\alpha$ and $\beta$ be coloured metric spaces and $A,B$ be constants. Our goal is to construct a rooted tree $T(\alpha, \beta, A,B)$ such that the following properties hold:\\
1. The tree $T(\alpha, \beta, A,B)$ is finitely branching and computable in $\alpha$ and $\beta$. In particular, if $\alpha$ and $\beta$ are computable then so is $T(\alpha, \beta, A,B)$. \\
2. For any $n$ one can effectively compute, with an oracle for $\alpha$ and $\beta$, the number of nodes in the tree at distance $n$ from the root.\\
3. The tree $T(\alpha, \beta, A,B)$ is infinite iff there exists an $(A,B)$-quasi-isometric map from $\alpha$ into $\beta$.\\
4. There is a bijection between $(A,B)$-quasi-isometric maps from $\alpha$ to $\beta$ and the set $[T(\alpha, \beta, A,B)]$ of all infinite paths of the tree $T(\alpha, \beta, A,B)$.

\smallskip

Informally, the nodes of  $T(\alpha, \beta, A,B)$ are finite partial functions that can potentially  be extended 
to $(A,B)$-quasi-isometric maps from $\alpha$ to $\beta$.  Formally:\\
1. The root of the tree $T(\alpha, \beta, A,B)$ is the empty set $\emptyset$, that is, nowhere defined partial function.\\
2. Let $x$ be a node of the tree $T(\alpha, \beta, A,B)$ constructed so far, and $f_x$ be the partial function associated with $x$.
We identify $x$ and $f_x$. We assume that
\begin{enumerate}
\item  $x$ is at distance $n$ from the root, and $f_x=\{(0,i_0), \ldots, (n-1, i_{n-1})\}$. 
\item For all $k,m$ from the domain of $f_x$ we have that $f_x$ is colour preserving and 
the $(A,B)$-quasi-isometry condition is satisfied:
\begin{center}{$
(1/A)\cdot d(k,m)-B < d(f_x(k),f_x(m))$\\
$< A\cdot d(k,m)+B.$}
\end{center}
\end{enumerate}
We list all the extension $\phi$ of $f_x$ such that the domain of $\phi$ is the set $\{0,1,\ldots, n\}$ and $\phi$ satisfies the 
the $(A,B)$-quasi-isometry condition. The children of $f_x$ will be all these functions $\phi$ extending $f_x$. 

\smallskip

In Part 2 of the construction the number of extensions $\phi$ of  $f_x$ is finite since these extensions must satisfy the 
quasi-isometry condition. Moreover, the number of these extensions is computed with an oracle for $\alpha$ and $\beta$.
If such extensions $\phi$  of $f_x$ do not exist, then $x$ is a leaf. 

\begin{lem} \label{L:Tree-lemma}
The tree  $T(\alpha, \beta, A,B)$ is infinite iff there is an $(A,B)$-quasi-isometric map from $\alpha$ into $\beta$.
There is a bijection between all infinite paths of $T(\alpha, \beta, A,B)$ and $(A,B)$-quasi-isometric maps from $\alpha$ to $\beta$.
\end{lem}

\begin{proof}
If $g:\alpha\rightarrow \beta$ is an $(A,B)$-quasi-isometric map from $\alpha$ to $\beta$ then the sequence
$\{g_i\}_{i\in \omega}$, where $g_i$ is the restriction of $g$ on the initial segment $\{0,\ldots, i\}$, is an infinite path through the
tree $T(\alpha, \beta, A,B)$. If $x_1, x_2, \ldots$ is an infinite path through the tree $T(\alpha, \beta, A,B)$ then $f$ defined as the limit of the sequence $f_{x_1}, f_{x_2}, \ldots$ is an $(A,B)$-quasi-isometry from $\alpha$ to $\beta$.
\end{proof}

The next theorem that solves the quasi-isometry problem. 
\begin{thm}\label{Thm:QI-complexity}
The following statements are true:
\begin{enumerate}
\item Given quasi-isometric strings $\alpha$ and $\beta$, there exists a quasi-isometry between $\alpha$ and $\beta$ computable in the halting set relative to $\alpha$ and $\beta$.
\item The quasi-isometry problem between computable strings
is a complete $\Sigma_3^0$-set.\qed
\end{enumerate}
\end{thm}
\begin{proof}

Assume that $\alpha$ and $\beta$ are $(A,B)$-quasi-isometric. The tree $T(\alpha, \beta, A,B)$ is infinite by the Lemma above.
Such a tree has a path computable in the halting set for the tree \cite{Soare}. The tree, as we constructed above, is computable in $\alpha$ and $\beta$. This proves the first part.

We prove the second part of the theorem. By Lemma \ref{L:Tree-lemma} computable strings $\alpha$ and $\beta$ are quasi-isometric if and only if there exists are constants $A, B$ such that the computable tree $T(\alpha, \beta, A,B)$ is infinite. This is a $\Sigma_3^0$-statement.
Hence, the quasi-isometry problem between computable strings $\alpha$ and $\beta$ is a $\Sigma_3^0$-set.
To prove that the problem is complete, we reduce a $\Sigma_3^0$-complete problem to the quasi-isometry problem. 
Below is an informal explanation of the reduction.

Let $W_0, W_1, \ldots$ be a standard enumeration of all c.e. sets. It is known that the set $Fin=\{i\mid W_i$ is finite$\}$ is a $\Sigma_3^0$-complete problem \cite{Rogers}. For each $W_i$ we need to construct $\alpha_i$ and $\beta_i$ two infinite strings such that $W_i$ is finite if and only if $\alpha_i$ and $\beta_i$ are quasi-isometric. For this we start enumerating $W_i$ by stages $1$, $2$, $\ldots$. At stage $0$, we consider $\alpha_{i,0}$ and $\beta_{i,0}$ finite prefixes that can be extended to a $(1,1)$-quasi-isometry and set $A_0=B_0=1$. By stage $s$ we will have a finite prefixes $\alpha_{i,s-1}$ of $\alpha_i$ and $\beta_{i,s-1}$ of $\beta_i$  built. If at stage $s$ the enumeration of $W_i$ outputs a new element, then we start extending $\alpha_s$ and $\beta_s$ so that the following holds. \  If $W_i$ never increases its size from stage $s$ on then $\alpha$ and $\beta$ are $(A_s, B_s)$-quasi-isometric but not $(A_{s-1}, B_{s-1})$-quasi-isometric, where $A_{s}>A_{s-1}$ and $B_s>B_{s-1}$. This can easily be achieved in two steps. In the first step, one  ensures that a large part of $\alpha$ contains a long consecutive sequence of elements of one colour, and  
the same positions of $\beta$ have another colour. This will guarantee that $\alpha_i$ and $\beta_i$ are not $(A_{s-1}, B_{s-1})$-quasi-isometric. In the second step, one ensures that those long sequences in the first step do not conflict with $(A_s, B_s)$-quasi-isometry and one can continue on extending $\alpha_{i,s}$ and $\beta_{i,s}$ so that they are $(A_s, B_s)$-quasi-isometric. Thus, if $W_i$ is infinite then there is an infinite sequence $s_1$, $s_2$, $\ldots$ of increasing stages at which the enumeration of $W_i$ increases the size of $W_i$. Hence, from stage $s_i$ on  $\alpha$ and $\beta$ are not $(A_{s_i}, B_{s_i})$-quasi-isometric. This implies that $W_i$ is infinite then $\alpha_i$ and $\beta_i$  are  not quasi-isometric.  
\end{proof}

\section{Asymptotic cones}\label{Cones}

Let $\mathcal F$ be a non-principal ultra-filter on $\omega$. Recall that a non-principal ultra-filter is a non-empty 
maximal 
subset of $P(\omega)$  that satisfies the following properties:
(1) No finite set belongs to $\mathcal F$. So,  $\mathcal F$ does not contain the empty set; 
(2) For all $A, B\in \mathcal F$ we have $A\cap B\in \mathcal F$; 
(3) For all $A, B\subseteq \omega$ if $A\in \mathcal F$ and $A\subseteq B$ then $B \in \mathcal F$. 

Every ultrafilter $\mathcal F$ has the following two properties: (1) For every set $A\subseteq \omega$, either $A\in \mathcal F$ or $\omega\setminus A\in \mathcal F$. (2) For all pairwise disjoint sets $A, B\subseteq \omega$ if $A\cup B\in \mathcal F$ then either $A\in \mathcal F$ or $B\in \mathcal F$. 

Let $\alpha\in \Sigma^{\omega}$ and  $s$ be a strictly increasing monotonic function on $\omega$ with $s(0)=1$. Call the mapping $s(n)$ a {\em scaling factor}.   Define the following sequence of metric spaces:
$$
X_{0,\alpha}=(\alpha, d_0), \ X_{1,\alpha}=(\alpha,d_1), \ldots, X_{n,\alpha}=(\alpha, d_n), \ldots  
$$
where  $d_n(i,j)=|i-j|/s(n)$. \ Informally, we move $\alpha$ away from us by scaling the metric down. For instance, 
in the metric space $X_{0,\alpha}$ the distance from $0$ to $s(n)$ equals $s(n)$, while in $X_{n, \alpha}$ the distance from $0$ to $s(n)$ is 
$1$. We assume that the domains of these metric spaces are disjoint pairwise:  $X_{i,\alpha}\cap X_{j,\alpha}=\emptyset$ for all $i\neq j$. 


Let ${\bf a}=(a_n)_{n\geq 0}$ be a sequence, where each $a_n\in X_{n, \alpha}$ for all $n$. Call the sequence 
{\em bounded} 
if there is a constant $L$ 
such that
 the set 
$\{n \mid d_n(0, a_n) < L \}\in \mathcal F$. Let ${\bf B}(\mathcal F, s)$ be the set of all bounded sequences. 
Say that two bounded sequences {\bf a} and {\bf b} are $\mathcal F$-equivalent, written 
${\bf a}\sim_{\mathcal F} {\bf c}$, if $\{n \mid d_n(a_n,b_n)\leq \epsilon\}\in \mathcal F$ for every $\epsilon$.

\begin{defn}\label{Dfn:Cone}
Define the {\bf asymptotic cone of $\alpha$}, written $Cone(\alpha, \mathcal F, s)$, with respect to the scaling function $s(n)$ and the ultra-filter $\mathcal F$ to be the factor set
$$
{\bf B}(\mathcal F, s) /\sim_{\mathcal F}
$$
equipped with the following metric $D$ and colour $C$:
\begin{enumerate}
\item $D({\bf a}, {\bf b} )=r$ iff 
$\{n \mid r-\epsilon \leq d_n(a_n,b_n)\leq r+\epsilon\}\in \mathcal F$ for all $\epsilon>0$. 
\item $C({\bf a})=\sigma$ iff  $\{n\mid a_n$ has colour $\sigma\}\in \mathcal F$. 
\end{enumerate}
\end{defn}
\noindent
This definition implies that some elements {\bf a} of the asymptotic cone might have several colours.

\begin{lem}\label{Lem:closed}
The set 
$\{ {\bf a} \in Cone(\alpha, \mathcal{F},s) \ | \ \sigma \in C({\bf a}) \}$  
is closed for each $\sigma \in \Sigma$,
\end{lem}
{\em Proof.}
Let {\bf a} be a limit point of the set above, ${\bf a_n}= (a_{nm})_{m \geq 0} \in Cone(\alpha, \mathcal{F},s)$ be a point of colour $\sigma$ such that $D({\bf a},{\bf a_n}) \leq 2^{-n}$. Also let $r = D({\bf 0},{\bf a})$, $r_n = D({\bf 0},{\bf a_n})$, and $i_m$ be a number such that $i_m \leq m$ and $|r-d_n(0,a_{i_mm})| \leq |r-d_n(0,a_{jm})|$ for all $j \leq m$. Then ${\bf b} = (a_{i_nn})_{n\geq0} \sim {\bf a}$ and $\sigma \in C({\bf b}) $: indeed, let $A_n = \{m \in \mathbb{N} \ |  \ |r_n -d_m(0,a_{nm})| \leq 2^{-n}\}$. Then for given $n$,
\begin{eqnarray*}
         && \{m \in \mathbb{N} \ |  \ |r -d_m(0,b_{m})| \leq 2^{-(n+1)}\} \\
&\supseteq& \{m \in \mathbb{N} \ |  \ |r -d_m(0,a_{i_mm})| \leq 2^{-(n+1)}\} \cup [n,\infty) \\
&\supseteq& A_{n+2} \cup [n,\infty) \in \mathcal{F}.
\end{eqnarray*}
Also $\{m \in \mathbb{N} \ |  \ |r -d_m(0,a_{m})| \leq 2^{-(n+1)}\} \in \mathcal{F}$, hence
\[
\{m \in \mathbb{N} \ |  \ d_m(a_{m},b_{m}) \leq 2^{-n}\} \in \mathcal{F}.\qed
\]

\begin{thm}
For eventually periodic word $\alpha=uv^{\omega}$  all  asymptotic cones $Cone(\alpha, \mathcal F, s)$ equal
the coloured metric
space $(\mathbb R_{\geq 0}; d, C)$ such that every real $r>0$ has colours of $v$ and the real $0$ has colours of both $u$ and $v$.
\end{thm}
\begin{proof} Let $\mathcal F$ be any filter and $s:\omega\rightarrow \omega$ be a monotonic function. It is easy to see that $0$ has all colours of both $u$ and $v$. Let $r \in \mathbb R$ and $r>0$. For any $\sigma$ present in $v$ there is a sequence $(a_n)_{n\geq 0}$ such that $a_n\in X_{n, \alpha}$, the colours of $a_n$ all are $\sigma$ and for every $\epsilon>0$ there is an $N_{\epsilon}$ such that for all $n>N_{\epsilon}$ we have \ 
$|d_n(a_n,0)-r|\leq \epsilon$. \ 
Indeed, for each $n \in \mathbb{N}$ such that $|u|/s(N_\epsilon) < r$ let $a_n$ be any element which has the colour $\sigma$ and satisfies $(|u|+i|v|)/s(n)\leq d_n(0,a_n) < (|u|+(i+1)|v|)/s(n)$, where $i$ is the unique integer which satisfies $(|u|+i|v|)/s(n)\leq r < (|u|+(i+1)|v|)/s(n)$. $|d_n(a_n,0)-r| \leq |v|/s(n)$. This inequality implies the claim. 
\end{proof}
\noindent
We now prove a theorem akin to a known result in geometric group theory stating that there are non-quasi-isometric groups that realise the same cones. For instance, in \cite{DP1} \cite{DP2} it is proved that all asymptotic cones of non-elementary hyperbolic group are isometric.

\begin{thm}
There are non-quasi-isometric strings $\alpha, \beta \in \{0,1\}^{\omega}$, 
a scale factor $s(n)$, and filter $\mathcal F$ such that  
$Cone (\alpha, \mathcal F, s)=Cone(\beta, \mathcal F, s)$ and $\alpha, \beta \in \mathcal X(1)$.
\end{thm}
\begin{proof}
Our scale factor is $s(n)=2^n$.  The string $\alpha$ is such that $\alpha(i)=1$ if and only if $i$ is a power of $2$.  
Let $\mathcal F$ be any ultra-filter. Consider the cone $Cone(\alpha, \mathcal F, s)$. It is note hard to see that
the cone coincides with the coloured metric space $(\mathbb R_{\geq 0}; d, C)$, where $d$ is the usual metric on reals, all reals have colour $0$, and a real $r$ has also colour $1$ iff $r$ is an integer power of $2$ or $r=0$. 

So, we need to construct $\beta$ and a filter $\mathcal F$ such that $\alpha$ and $\beta$ are not quasi-isometric but the cone
$Cone(\beta, \mathcal F, s)$ coincides with the above coloured metric space $(\mathbb R_{\geq 0}; d, C)$.
Let $\beta$ be a string of the form
$$
0^{m_0} 1 0^{m_1} 10^{m_2} 10^{m_3}\ldots 0^{m_k} 1 0^{m_{k+1}}\ldots
$$
such that the following properties hold:
\begin{enumerate}
\item The positions where $\beta$ contains $1$ are powers of two, and let us list these positions as the  sequence $2^{n_0}, 2^{n_1}, 2^{n_2}, \ldots$.
\item For each $i\in \omega$ we have $2n_i< n_{i+1}$.
\item The strings $\alpha$ and $\beta$ are not quasi-isometric.
\end{enumerate}
Now we construct our filter $\mathcal F$. Let $X$ be the set 
$$
\{n_0,n_0+1,\ldots, 2n_0, \ldots, n_k, n_k+1, \ldots, 2n_k, \ldots\}.
$$
Call the sequences of the form $n_k, n_k+1, \ldots, 2n_k$ {\em blocks} of the set $X$. Note that the lengths of the blocks is unbounded. For each integer $i$, consider the set $X+i=\{x+i\mid x\in X, \  x+i\geq 0 \}$. 
Since the sizes of the blocks is unbounded,  the collections of sets is a bases of a filter:
$$
X, \  X+1, \ X-1, \ X+2, \ X-2,  \ldots, X+i, \  X-i, \ldots
$$
Let $\mathcal F$ be the ultrafilter that contains the collection.

Consider the cone $Cone(\beta, \mathcal F, s)$. We can view the domain of this cone as the set $\mathbb R_{\geq 0}$. 
It is not too hard to note that every real $r$ gets colour $0$ in the cone $Cone(\beta, \mathcal F, s)$. Now we need to show 
that $r$ gets colour $1$ if and only if $r$ is an integer power of $2$. It suffices to show that if $r$ is an integer power of $2$ then
$r$ has colour $1$. But, this is implied by the fact that the sets $X+i$ are in $\mathcal F$. Indeed, assume that $r$ is of the form $2^i$. Then, since $X+i \in \mathcal F$, from the definition of $Cone(\beta, \mathcal F, s)$ we see that $r$ has colour $1$.

To finish the proof, we just need to show that we can select the sequence $n_0, n_1, \ldots$ such that $\alpha$ and $\beta$ are not quasi-isometric. To see this, let $n_{k+1} = 2n_k +2$: then we can show that $\alpha$ and $\beta$ are not quasi-isometric in the same manner as the proof of Theorem \ref{P:infinite-chain}.
\end{proof}

In contrast to the theorem above, we show that the same coloured metric space can produce two asymptotic cones that are
not quasi-isometric. This is similar to the result in geometric group theory where one group can realise two non homeomorphic asymptotic cones  \cite{Thomas}.

\begin{thm}
There is a sequence $\alpha$, scaling factors $s_0$ and $s_1$ so that for all ultrafilters $\mathcal{F}$  cones  
$Cone(\alpha, \mathcal F, s_0)$ and $Cone(\alpha, \mathcal F, s_1)$ are not quasi-isometric.
\end{thm}

\begin{proof}
We construct $\alpha$ and two scaling factors $s_0$ and $s_1$ such that in the asymptotic cone $Cone(\alpha, \mathcal F, s_0)$ all reals $r>1$ have colour $0$, and in  the asymptotic cone $Cone(\alpha, \mathcal F, s_1)$ all reals $r>1$ have colour $1$ only. These two metric spaces are clearly not quasi-isometric. 


Define $p_n$ and $q_n$, inductively, by: $p_1=q_1=1$, and 
\begin{eqnarray*}
p_{n+1} &=& n \sum_{i=1}^n (p_i + q_i) , \\
 q_{n+1} & = & n \biggl( \sum_{i=1}^n (p_i + q_i) + p_{n+1} \biggr)
\end{eqnarray*}
Now consider the string $\alpha = 0^{p_1} 1^{q_1} 0^{p_2} 1^{q_2} \ldots$, and define the following two scale functions: 
\begin{itemize}
\item $s_0(n) = \sum_{i=1}^n (p_i + q_i)$, and 
\item $s_1(n) = \bigl( \sum_{i=1}^n (p_i + q_i) + p_{n+1} \bigr)$.
\end{itemize}
Also let $d_{j,n}(x,y)= |y-x|/s_j(n) \ (j=0,1)$. Then all points $a$ in $(\alpha, d_{j,n})$ such that $1 < d_{j,n}(0,a) < n $ have the colour $j$. Hence for any $r>1$, the colour of $r$ is $0$ in $Cone(\alpha, \mathcal F, s_0)$ and $1$ in $Cone(\alpha, \mathcal F, s_1).$ 
\end{proof}

The next theorem shows that the asymptotic cones of all Martin-L\"of random strings \cite{Nies} coincide when the scaling factor is a computable function.
\begin{thm}\label{Thm:RandomCone}
If $\alpha$ is Martin-L\"of random, then for all computable scaling factors $s$ and ultra-filters $\mathcal F$, 
the asymptotic cone $Cone(\alpha, \mathcal F, s)$ coincides with the space $(\mathcal R_{\geq 0}; d, C)$, where 
every real has all colours from $\Sigma$. 
\end{thm}

\begin{proof}

We assume that $\Sigma=\{0,1\}$. For given $r \in \mathbb{R}_{\geq0}$, assume the following:
\[
\exists A \in \mathcal{F} \forall \epsilon \forall^\infty n \in A \exists r' \in \alpha_n [r-\epsilon \leq r' \leq r+\epsilon \land C_n(r') = 0].
\]
Then we have a sequence $(a_n)$ such that the subsequence $\{a_n \ | \ n \in A\}$ converges to $r$ and all elements of the subsequence are coloured by $0$, and hence $Cone(\alpha, \mathcal F, s)$ has the colour $0$ at $r$.

Taking the contraposition of this fact and letting $A = \mathbb{N}$, we have the following:
if $Cone(\alpha, \mathcal{F}, s)$ does not have the colour $0$ at $r$ (by virtue of Lemma \ref{Lem:closed} we can assume that  $r \in \mathbb{Q}_{\geq 0}$), then for some $\epsilon >0$ we have infinitely many $n \in \mathbb{N}$ such that all points $r'$ of $\alpha_n$ in an interval $[r-\epsilon, r+\epsilon]$ has the colour 1.

 Now assume such an $r$ exists and consider the set:
\[
G_n = \{X \in 2^\omega \ | \  s(n)(r-\epsilon) \leq i \leq  s(n)(r+\epsilon) \Rightarrow X(i)=1\}
\]
In the Cantor space $\{0,1\}^{\omega}$, the sequence of open sets $(G_n)$ is uniformly computably enumerable and  the 
sum of the measures of $G_n$ is bounded. Such sequences are called Solovay tests.  This Solovay test fails the string $\alpha$,
that is, $\alpha\in G_n$ for infinitely many $n$ (cf. Definition 3.2.18 in \cite{Nies}). It is known that falling Solovay tests is equivalent failing Martin-L\"of tests (cf.  Proposition 3.2.19 in  \cite{Nies}). 
This shows that $\alpha$ is not random. 
\end{proof}

A culmination of asymptotic cones construction is that $[\alpha]=[\beta]$ implies 
bi-Lipschitz equivalence between  $Cone(\alpha, \mathcal F, s)$ and $Cone(\beta, \mathcal F, s)$. The 
proof is standard 
with the difference that our metric spaces are coloured.  However, since our base metric space is a linear order  $\mathbb \omega$, the conclusion of our theorem is a little stronger as it implies order preservance.  

\begin{thm} \label{Thm:Cones}
If strings $\alpha$ and $\beta$ have the same large scale geometry  then  there are colour-preserving and order preserving  homeomorphisms 
\begin{center}{
$H:Cone(\alpha, \mathcal F, s)\rightarrow Cone(\beta, \mathcal F, s)$ and\\
$G:Cone(\beta, \mathcal F, s) \rightarrow Cone(\alpha, \mathcal F, s)$,}
\end{center} 
constants $C_H$, $C_G$ such that for all \ ${\bf a}, {\bf b} \in Cone(\alpha, \mathcal F, s)$ and  \ ${\bf c}, {\bf d} \in Cone(\beta, \mathcal F, s)$ we have:
\begin{center}{
$(1/C_{H} )\cdot D({\bf a}, {\bf b})\leq D(H({\bf a}), H({\bf b}))\leq C_{H}\cdot D({\bf a}, {\bf b})$\\ and \\
\  \  \  $(1/C_{G}) \cdot D({\bf c}, {\bf d})\leq D(G({\bf c}), G({\bf d}))\leq C_{G}\cdot D({\bf c}, {\bf d})$.\qed}
\end{center}
\end{thm}

\section*{Acknowledgment}
This work was supported by JST ERATO Grant Number JPMJER1603, Japan. The authors thank Kazushige Terui for discussions. The first author thanks JSPS and RIMS of Kyoto University for support.

\end{document}